\documentclass[a4paper,USenglish]{lipics-v2021}
\nolinenumbers

\usepackage[WOamssymb,WOsilence,WOdmath,WOphantomas,WOmparhack,WOprobsubsuperscripts,WOaddbibresource]{auxlib}

\usepackage{xspace}
\usepackage{float}

\usepackage{thm-restate}

\setquotestyle{english}
\MakeOuterQuote{"}

\usepackage{tikz}
\usetikzlibrary{decorations.pathreplacing}

\makeatletter
\let\c@author\relax
\makeatother
\usepackage[
    isbn=false,
    sortcites,
    maxbibnames=99,
    maxcitenames=2,
    maxalphanames=4,
    minalphanames=3,
    style=numeric,
    url=false,
    giveninits=true,
    backend=biber,
]{biblatex}

\AtEveryBibitem{\ifentrytype{inproceedings}{\clearname{editor}\clearlist{address}\clearlist{location}}{}}
\AtEveryBibitem{\clearlist{language}}

\makeatletter
\newcommand{\DeclareMathActive}[2]{%
  \expandafter\edef\csname keep@#1@code\endcsname{\mathchar\the\mathcode`#1 }
  \begingroup\lccode`~=`#1\relax
  \lowercase{\endgroup\def~}{#2}%
  \AtBeginDocument{\mathcode`#1="8000}%
}

\newcommand{\std}[1]{\csname keep@#1@code\endcsname}
\patchcmd{\newmcodes@}{\mathcode`\-\relax}{\std@minuscode\relax}{}{\ddt}
\AtBeginDocument{\edef\std@minuscode{\the\mathcode`-}}
\makeatother

\DeclareMathActive{O}{\operatorname{\std{O}}}

\makeatletter
\def\?#1{}
\def\whp{w.h.p\@ifnextchar.{.\?}{\@ifnextchar,{.}{\@ifnextchar){.}{\@ifnextchar:{.:\?}{.\ }}}}}
\def\Whp{W.h.p\@ifnextchar.{.\?}{\@ifnextchar,{.}{.\ }}}
\makeatother

\renewcommand\MoveEqLeft[1][1.5]{\kern #1em  &   \kern -#1em}

\newcommand{\bc}[1]{\left( #1 \right)}
\newcommand{\brk}[1]{\left[ #1 \right]}
\newcommand{\cbc}[1]{\left\lbrace #1 \right\rbrace}
\newcommand{\Erw}{\mathbb{E}}
\createcustomnote{dcm}{backgroundcolor=AUXLlightblue, prefix=Dominik}
\createcustomnote{mhk}{backgroundcolor=AUXLlightred, prefix=Max}
\createcustomnote{malin}{backgroundcolor=AUXLgreen, prefix=Malin}
\createcustomnote{geb}{backgroundcolor=violet!50, prefix=Oli}
\createcustomnote{todo}{backgroundcolor=orange}

\renewcommand\P[2][]{P_{#2}^{\text{\normalfont#1}}}
\def\Pl{\P[]}
\def\Pn{\P[]}

\def\Bin{\mathrm{Bin}}

\let\eps\varepsilon

\def\protocol#1{\textsc{#1}\xspace}

\author{Petra Berenbrink}{Universität Hamburg, Germany}{petra.berenbrink@uni-hamburg.de}{}{DFG FOR 2975.}

\author{Amin Coja-Oghlan}{TU Dortmund University, Germany}{amin.coja-oghlan@tu-dortmund.de}{}{DFG FOR 2975.}

\author{Oliver Gebhard}{TU Dortmund University, Germany}{oliver.gebhard@tu-dortmund.de}{}{DFG CO 646/3.}

\author{Max Hahn-Klimroth}{TU Dortmund University, Germany}{maximilian.hahnklimroth@tu-dortmund.de}{0000-0002-3995-419X}{DFG FOR 2975.}

\author{Dominik Kaaser}{TU Hamburg, Germany}{dominik.kaaser@tuhh.de}{0000-0002-2083-7145}{}

\author{Malin Rau}{Universität Hamburg, Germany}{malin.rau@uni-hamburg.de}{}{DFG FOR 2975.}

\authorrunning{P.\ Berenbrink, A.\ Coja-Oghlan, O.\ Gebhard, M.\ Hahn-Klimroth, D.\ Kaaser, and M.\ Rau} 

\Copyright{Petra Berenbrink, Amin Coja-Oghlan, Oliver Gebhard, Max Hahn-Klimroth, Dominik Kaaser, and Malin Rau} 


\hideLIPIcs

\EventEditors{}
\EventLongTitle{}
\EventShortTitle{}
\EventAcronym{}
\EventYear{}
\EventDate{}
\EventLocation{}
\EventLogo{}
\SeriesVolume{}
\ArticleNo{}

\title{On the Hierarchy of Distributed Majority Protocols}

\begin{document}

\maketitle

\begin{abstract}
We study the \emph{Consensus} problem among $n$ agents, defined as follows.
Initially, each agent holds one of two possible opinions.
The goal is to reach a consensus configuration in which every agent shares the same opinion.
To this end, agents randomly sample other agents and update their opinion according to a simple update function depending on the sampled opinions.

We consider two communication models: the \emph{gossip model} and a variant of the \emph{population model}.
In the gossip model, agents are activated in parallel, synchronous rounds.
In the population model, one agent is activated after the other in a sequence of discrete time steps.
For both models we analyze the following natural family of majority processes called $j$-Majority: when activated, every agent samples $j$ other agents uniformly at random (with replacement) and adopts the majority opinion among the sample (breaking ties uniformly at random).
As our main result we show a hierarchy among majority protocols: $(j+1)$-Majority (for $j > 1$) converges \emph{stochastically faster} than $j$-Majority for any initial opinion configuration.
In our analysis we use Strassen's Theorem to prove the existence of a coupling.
This gives an affirmative answer for the case of two opinions to an open question asked by Berenbrink et al.\ [2017].
\end{abstract}

\section{Introduction}
In this paper we consider consensus protocols in a distributed system consisting of $n$ identical, anonymous agents. Initially every agent has one of $k$  opinions  and 
the goal is that all agents agree on the same opinion.
Reaching \emph{consensus} is a fundamental task in distributed computing with a multitude of applications including fault tolerance in distributed sensor array, clock synchronization, control of autonomous robots, or blockchains.
In computational sciences, consensus protocols model, e.g., dynamic particle systems or biological processes.
In social sciences, consensus protocols have been studied in the context of opinion formation processes among social interaction systems.
See~\cite{DBLP:journals/sigact/BecchettiCN20} for a quite recent survey including references and further applications.

We study the simple and well-known class of \emph{$j$-Majority} protocols \cite{DBLP:journals/dc/BecchettiCNPST17, DBLP:conf/podc/GhaffariL18, DBLP:conf/podc/BerenbrinkCEKMN17} in the gossip model \cite{DBLP:conf/stoc/Censor-HillelHKM12, DBLP:conf/soda/BecchettiCNPS15, DBLP:journals/sigact/BecchettiCN20} where all agents are activated in parallel rounds. In turn, every agent $u$ considers the opinions of $j$ agents $v_1, \dots, v_j$ sampled uniformly at random (with replacement).
It then adopts the majority opinion among the sampled opinions, breaking ties uniformly at random. We are interested in the time it takes until the protocol \emph{converges} to such a consensus configuration. 
Setting $j = 1$ yields the so-called \protocol{Voter} process \cite{DBLP:conf/icalp/BerenbrinkGKM16}. A variant of $2$-Majority with lazy tie-breaking is known as \protocol{two-sample voting} \cite{DBLP:conf/icalp/CooperER14} or the \protocol{TwoChoices} process \cite{DBLP:conf/podc/GhaffariL18}, and the $3$-Majority dynamics is analyzed in \cite{DBLP:journals/dc/BecchettiCNPST17}.

The main idea of majority process with $j>1$ is to speed up the convergence time. For the \protocol{voter} process the convergence time is linear in $n$
\cite{DBLP:conf/icalp/BerenbrinkGKM16} (independent of the number of initial opinions),  whereas the convergence time of $3$-Majority is $O(k \log n)$ for $k$ $(= o(n))$ initial opinions \cite{DBLP:conf/podc/GhaffariL18}.
In  \textcite{DBLP:conf/podc/BerenbrinkCEKMN17} the authors compare the \protocol{TwoChoices} process to $3$-Majority.
They show a stochastic dominance of the convergence time of $3$-Majority over the convergence time of \protocol{Voter} and \protocol{TwoChoices}, assuming $k$ initial opinions.
For $j$-Majority, they conjecture a \emph{hierarchy} of protocols (see Conjecture 6.1 in \cite{DBLP:conf/podc/BerenbrinkCEKMN17}). In particular, they ask whether one can couple $j$-Majority and $(j+1)$-Majority for $j \in \mathbb{N}$ such that $(j + 1)$-Majority is stochastically faster than $j$-Majority.

\medskip

In this paper, we settle the matter for the case of $k = 2$ opinions and prove the existence of such a hierarchy of majority protocols. 
Intuitively, this establishes that the processes converge \emph{faster} (or at least \emph{equally fast}) for larger values of $j$.
Let $T_j$ be the random variable for the convergence time of $j$-Majority.
We formally prove that $T_{j+1}$ stochastically minorizes $T_{j}$, written $T_{j+1} \preceq T_{j}$, assuming both processes start in the same configuration.
Formally, we show that $\Prob{T_{j+1} \geq t} \leq \Prob{T_{j} \geq t}$ for any $t \in \N$. Our main technical contribution is the formal proof of this stochastic dominance. 
Our proof has its foundations in quite natural observations regarding the transition properties of the $j$-Majority processes. Similar results for individual steps of the process have been shown, e.g., in \cite{DBLP:journals/dc/FraigniaudN19}.
However, formally proving and maintaining the stochastic dominance over all possible configurations requires a lot of care, and to the best of our knowledge, our result is the first proof of stochastic dominance that covers the entire execution of $j$-Majority for all $j \in \mathbb{N}$ in the setting with two opinions.
To motivate the obstacles we have to overcome, observe that the process is influenced by opposing "forces".
Specifically, an agent from the minority opinion must be selected to interact with more than $j/2$ agents from the majority opinion.
The former becomes less likely with increasing majority, while the latter becomes more likely with increasing majority.
In the proof we carefully show that these forces balance out in a favorable manner.

In addition, we also prove the stochastic dominance in a \emph{sequential model} where one agent after the other is randomly activated. Note that this sequential model is a variant of the prominent \emph{population model} \cite{DBLP:journals/dc/AngluinADFP06}, where in each time step a \emph{pair} of agents interact.
Finally, we show for $3$-Majority in the sequential model an asymptotically optimal bound on the convergence time of $O(n \log n)$ activations. In \cite{DBLP:conf/podc/GhaffariL18} the authors show the same result for $3$-Majority in the gossip model.
Our theoretical findings are complemented by empirical results. We simulate $j$-Majority processes for various values of $j$ and large numbers of agents ranging from $n = 10^2$ to $n = 10^8$.

\subsection{Related Work}

\paragraph{Consensus in the Gossip Model}
A simple and natural consensus process is the so-called \protocol{Voter} process \cite{DBLP:journals/iandc/HassinP01, DBLP:journals/networks/NakataIY00, DBLP:conf/podc/CooperEOR12, DBLP:conf/icalp/BerenbrinkGKM16, DBLP:conf/soda/KanadeMS19}
where every agent adopts the opinion of a single, randomly chosen agent in each round.
The expected convergence time of \protocol{Voter} is at least linear \cite{DBLP:conf/icalp/BerenbrinkGKM16}.
In order to speed up the process, two related protocols have been proposed, namely the \protocol{TwoChoices} process \cite{DBLP:conf/podc/ElsasserFKMT17,DBLP:conf/icalp/CooperER14,DBLP:conf/wdag/CooperERRS15,DBLP:conf/wdag/CooperRRS17} and the \protocol{3-Majority} dynamics \cite{DBLP:journals/dc/BecchettiCNPST17, DBLP:conf/podc/GhaffariL18, DBLP:conf/podc/BerenbrinkCEKMN17}.
In both processes, each agent $u$ takes three opinions and updates its opinion to the majority among the sample.
In the \protocol{TwoChoices} process, $u$ takes its own opinion and samples two opinions u.a.r. Ties are broken towards $u$'s own opinion.
In the \protocol{3-Majority} dynamics, $u$ samples three opinions u.a.r.\ breaking ties randomly.
In \cite{DBLP:conf/podc/GhaffariL18} the authors consider arbitrary initial configurations in the gossip model.
They show that \protocol{TwoChoices} with $k = O(\sqrt{n / \log n})$ and \protocol{3-Majority} with $k = O({n^{1/3} / {\log n}})$ reach consensus in $O({k\cdot \log n})$ rounds, improving a result by \textcite{DBLP:journals/dc/BecchettiCNPST17}.
For arbitrary $k$, they show that \protocol{3-Majority} reaches consensus in $O(n^{2/3} \log^{3/2} n)$ rounds \whp, improving a result by \textcite{DBLP:conf/podc/BerenbrinkCEKMN17}.

\Textcite{DBLP:conf/soda/SchoenebeckY18} consider a generalization of multi-sample consensus protocols on complete and Erdős-Rényi graphs for two opinions.
Their probabilistic model covers various consensus processes, including $j$-Majority, by using a so-called \emph{update rule}, a function $f \colon [0,1] \rightarrow [0,1]$.
In each round, every agent $u$ adopts opinion $a$ with probability $f(\alpha(u))$ for some function $f$, where $\alpha(u)$ is the fraction of neighbors of agent $u$ that have opinion $a$.
Depending on certain natural properties on $f$, they analyze the convergence time for complete graphs and Erdős-Rényi graphs.

Another related process is the \protocol{MedianRule}~\cite{DBLP:conf/spaa/DoerrGMSS11}
where in each round every agent adopts the median of its own opinion and two sampled opinions, assuming a total order among opinions.
It reaches consensus in $O({\log k \log\log n + \log n})$ rounds \whp.
For two opinions the \protocol{MedianRule} is equivalent to the \protocol{TwoChoices} process, and their analysis is tight.
For the case of $k > 2$ opinions we remark that assuming a total order among the opinions is a strong assumption that is not required by any of the other protocols.

Finally, considerate amount of work has been spent on analyzing the so-called undecided state dynamics introduced by \textcite{DBLP:journals/dc/AngluinAE08}.
The basic idea is that whenever two agents with different opinions interact, they lose their opinions and become \emph{undecided}, and undecided agents adopt the first opinion they encounter.
\Textcite{DBLP:conf/mfcs/ClementiGGNPS18} study the undecided state dynamics in the gossip model.
They consider two opinions and show that the protocol reaches consensus in $O({\log n})$ rounds \whp.
If there is a so-called \emph{bias} of order $\Omega(\sqrt{n \log n})$, the initial plurality opinion prevails.
The (additive) bias is the difference between the numbers of agents holding either opinion.
\Textcite{DBLP:conf/soda/BecchettiCNPS15} analyze the undecided state dynamics for $k = O{{(n/\log n)}^{1/3}}$ opinions and show a convergence time of $O({k \cdot \log n})$ rounds \whp.
\Textcite{DBLP:conf/podc/GhaffariP16a, DBLP:conf/icalp/BerenbrinkFGK16, DBLP:journals/corr/abs-2103-10366} consider a synchronized variant that runs in phases of length $\Theta({\log k})$.
Agents can become undecided only at the start of such a phase and use the rest of the phase to obtain a new opinion.
These synchronized protocols achieve consensus in $O({\log^2 n})$ rounds \whp and can be further refined using more sophisticated synchronization mechanisms.

\paragraph{Majority and Consensus in the Population Model}
In \emph{exact} majority the goal is to identify the majority among two possible opinions, even if the bias is as small as only one~\cite{DBLP:journals/siamco/DraiefV12, DBLP:conf/icalp/MertziosNRS14, DBLP:conf/podc/AlistarhGV15, DBLP:conf/nca/MocquardAABS15, DBLP:journals/dc/DotyS18, DBLP:conf/soda/AlistarhAEGR17, DBLP:conf/soda/AlistarhAG18, DBLP:conf/podc/BilkeCER17, DBLP:conf/podc/KosowskiU18, DBLP:conf/wdag/BerenbrinkEFKKR18, DBLP:journals/dc/BerenbrinkEFKKR21, DBLP:conf/podc/NunKKP20, DBLP:journals/corr/abs-2106-10201}.
The best known protocol by \textcite{DBLP:journals/corr/abs-2106-10201} solves exact majority with $O(\log n)$ states and $O(\log n)$ parallel time, both in expectation and \whp.
This is optimal: it takes at least $\Omega({n \log n})$ interactions until each agent interacts at least once, and any majority protocol which stabilizes in expected ${n^{1-\Omega({1})}}$ parallel time requires at least $\Omega({\log{n}})$ states (under some natural conditions, see \cite{DBLP:conf/soda/AlistarhAG18}).

\emph{Approximate} majority is easier: a simple 3-state protocol \cite{DBLP:journals/dc/AngluinAE08, DBLP:journals/nc/CondonHKM20} reaches consensus \whp in $O({\log n})$ parallel time and correctly identifies the initial majority \whp if an initial bias of order $\Omega{(\sqrt{n \log n})}$ is present.
\Textcite{DBLP:journals/nc/CondonHKM20} also consider a variant of the 3-Majority process in (a variant of) the gossip model where three randomly chosen agents interact.
They show a parallel convergence time of $O(k \log n)$ \whp, provided a sufficiently large initial bias is present.
Furthermore, \textcite{DBLP:conf/podc/KosowskiU18} mention a protocol which determines the exact majority in $O({\log^2 n})$ parallel time \whp using only constantly many states.

Less is known about population protocols that solve consensus among more than two opinions.
One line of research considers only the required number of states to eventually identify the opinion with the largest initial support correctly.
For this problem, \textcite{DBLP:conf/ciac/NataleR19} show a lower bound of $\Omega(k^2)$ states via an indistinguishability argument.
The currently best known protocol uses $O(k^6)$ states if there is an order among the opinions and $O({k^{11}})$ states otherwise~\cite{DBLP:conf/opodis/GasieniecHMSS16}.
Sacrificing the strong guarantees of always-correct exact plurality consensus, \textcite{DBLP:journals/corr/abs-2103-10366} achieve \emph{approximate consensus} in $O(\log^2 n)$ parallel time \whp using only $O(k \log n)$ states.
If there is an initial bias of order $\Omega({\sqrt{n \log n}})$, the initial plurality opinion wins \whp.
In \cite{DBLP:conf/podc/BankhamerEKK20} another variant of the population model is considered where agents are activated by random clocks.
At each clock tick, every agent may open communication channels to constantly many other agents chosen uniformly at random or from a list of at most constantly many agents contacted in previous steps.
In this model, opening communication channels is subject to a random delay.
The authors show that consensus is reached by all but a $1/\poly\log n$ fraction of agents in $O({\log\log_\alpha k \log k + \log\log n})$ parallel time \whp, provided a sufficiently large bias is present.

\subsection{Models and Results}

\paragraph{Gossip Model}
In the gossip model \cite{DBLP:conf/stoc/Censor-HillelHKM12, DBLP:conf/soda/BecchettiCNPS15, DBLP:journals/sigact/BecchettiCN20} all agents are activated simultaneously in synchronous rounds.
In each round every agent $u$ opens a communication channel to $j$ agents $v_1, \dots, v_j$ chosen independently and uniformly at random with replacement.
 (For simplicity we also allow that $v_i = u$ and assume that the $v_i$ are sampled with replacement.)
The running time (or \emph{convergence time}) of a majority protocol is measured in the numbers of rounds until all agents agree on the same opinion.

\paragraph{Sequential Model}
The population model was introduced by \textcite{DBLP:journals/dc/AngluinADFP06} to model systems of resource limited mobile agents that perform a computation via a sequence of pairwise interactions.
We consider a variant where in each time step one agent $u$ is chosen uniformly at random to interact with $j$ randomly sampled agents $v_1, \dots, v_j$. (As before, we do not rule out that $u = v_i$ for some $i$).
When $u$ is activated it updates its opinion according to the random sample.
The running time is measured in the number of interactions.
To allow for a comparison with the (inherently) parallel gossip model, the so-called \emph{parallel time} is defined as the number of interactions divided by the number of agents $n$.
Note that our processes do not \emph{halt}: agents do not know that consensus has been achieved (see also the impossibility result in \cite{DBLP:conf/podc/DotyE19}).

\paragraph{\textit{j}-Majority Processes}
In the following we use $P_j$ to denote the $j$-Majority process.
When executing process $P_j$, the system transitions through a sequence of \emph{configurations} $\left(C_t\right)_{t \in \N_0}$.
At time $t \in \mathbb{N}_0$ the \emph{configuration} $C_t \in \set{a,b}^n$ assigns each agent an opinion in $\set{a,b}$.
In our analysis we are interested in the number of agents with majority opinion. 
We will always assume w.l.o.g.\ that $a$ is the majority opinion and we denote a \emph{state} $X_t$ as the number of agents with majority opinion in configuration $C_t$.
The configuration $C_0$ at time $0$ is called the \emph{initial configuration} and the corresponding state $X_0$ is called the \emph{initial state}.
The \emph{convergence time} $T_j(C_0)$ is defined as the first time where all agents have the same opinion when starting process $P_j$ in initial configuration $C_0$.
Note that the convergence time in the complete graph only depends on the number of agents with majority opinion since two nodes with the same opinion are not distinguishable. 
Hence we write $T_j(X_0)$ in the following.
Formally, $T_j(X_0)$ is a stopping time defined as $T_j(X_0) = \min\set{t \in \mathbb{N}_0 | X_t =n}$.

Whenever an agent $u$ is activated in process $P_j$, it samples a set of $j$ agents uniformly at random and updates its opinion according to the following rules.
\begin{itemize}[nosep]
\item In process $P_{2j}$, the agent $u$ samples $2j$ agents from $V$ with replacement.
It then adopts the majority opinion among the sample, breaking ties uniformly at random.
\item In process $P_{2j+1}$, the agent $u$ samples $2j+1$ agents with replacement.
It then adopts the majority opinion among the sample.
\end{itemize}
Note that tie-breaking is not required in process $P_{2j+1}$: we are guaranteed to have a clear majority since we have $k = 2$ opinions.

\paragraph{Stochastic Dominance}
Before we formally present our result, it remains to define stochastic dominance.
\begin{definition*}[Stochastic Dominance]
Let $\cE$ be a Polish space endowed with a partial ordering $\leq_{\cE}$. Let $\mu, \nu \in \cP (\cE)$ be probability measures on $\cE$. If, for every $x \in \cE$, we have
\begin{align*}
    \mu \bc{ \cbc{ y \in \cE: y \geq_{\cE} x } } \geq \nu \bc{ \cbc{ y \in \cE: y \geq_{\cE} x } },
\end{align*}
we say that $\mu$ stochastically dominates $\nu$. In this case we also say that $\mu$ \emph{majorizes} $\nu$ (written as $\nu \succeq \nu$) or $\nu$ \emph{minorizes} $\mu$ ($\nu \preceq \nu$).
\end{definition*}
We now formally state our main result which applies for both communication models, the gossip model and the sequential model. 

\medskip

\noindent\hspace{1.5em}\colorbox{lightgray!25}{\parbox{\textwidth-3em}{
\begin{restatable}[Main Result]{theorem}{MainTheorem}
\label{thm:main-result}
Let $  T_j(X_{0})$ be the convergence time of process $P_{j}$ with initial state $X_0$ in either the gossip model or the sequential model.
Then
\begin{align*}
      T_{j+1}(X_{0}) \preceq   T_{j}(X_{0}) \quad \text{for any } j > 1.
\end{align*}
Furthermore, for all $j > 1$,
\begin{align*}
    \Ex{  T_{2j+2}(X_{0}) } = \Ex{  T_{2j+1}(X_{0}) } \leq \Ex{  T_{2j}(X_{0}) }.
\end{align*}\end{restatable}
}}

\medskip

\noindent In our second result we show that $3$-Majority $P_3$ converges in $O(n \log n)$ time w.h.p.\footnote{The expression \emph{with high probability (w.h.p.)} refers to a probability of $1 - n^{-\LDAUOmega{1}}$.}\,
To the best of our knowledge, this is the first analysis of $3$-Majority with sequential updates.
Our proof is similar to the proof by \textcite{DBLP:journals/nc/CondonHKM20} for the convergence time of approximate majority in tri-molecular chemical reaction networks.
We emphasize that \cref{thm:main-result} implies that all $j$-Majority processes with $j > 3$ converge in $O(n \log n)$ time \whp.

\begin{theorem}
\label{thm:convergencetime}
Let $T_3(X_0)$ be the convergence time of the $3$-Majority process $P_3$ in the sequential model with initial configuration $X_0$.
\begin{enumerate}[nosep]
\item It holds that $  T_3(X_0) \leq O \bc{ n \log n}$ \whp.
\item If $X_0 \geq n/2 + \zeta\sqrt{n\log n}$ for some sufficiently large constant $\zeta > 0$ then the initial majority opinion wins \whp.
\end{enumerate}
\end{theorem}
We remark that the convergence time of $O(n \log n)$ is asymptotically tight.
Indeed, for any number of time steps in $o(n \log n)$ there is a constant probability that two agents with opposing opinions are not activated even once.

\section{Analysis}

In this section we formally prove our theorems. We prove
\cref{thm:main-result} in \cref{sec:sequential} and \cref{sec:gossip} for the sequential model and the gossip model, respectively.  \Cref{thm:convergencetime} is then shown in \cref{sec:3Majority}. All technical details for the rigorous proofs can be found in \cref{apx:appendix}.

\subsection{Sequential Model}\label{sec:sequential}

We start our analysis with a comparison of one step of the processes $\Pl{j}$ and $\Pn{j+1}$ at time $t$ when starting in an identical state $X_{t}$. 
We are able to express the differences in the probabilities of increasing the majority opinion, decreasing it or remaining in the same state for the both processes. 
To this end, we visualize a possible coupling by a decision tree that incorporates all the different possibilities.
We will observe that, within this one step, we can couple the both processes such that the supposedly faster process increases the majority opinion with probability one if the supposedly slower process increases this opinion. 
This coupling will be guaranteed by an application of \emph{Strassen's Theorem}.

The proof of the main result will be conducted inductively. 
We start both processes in the same initial state and assume that there is a majority opinion $a$. 
Now, the aforementioned coupling ensures that, after the first step, the supposedly faster process will have at least as many agents of opinion $a$ than the supposedly slower process. 
Now, we show a kind of \emph{monotony} in the studied processes. 
Assume we have two instances of the same process, one in state $X_t = s$ and one in state $X'_t = s'$ where $X_t, X'_t$ denote the number of agents with opinion $a$ after $t$ steps. 
If $s > s'$, then the random variable $X_{t+1}$ will stochastically dominate $X'_{t+1}$, formally $X'_{t+1} \preceq X_{t+1}$. 
This observation is crucial. 
It allows us to show that in the second step, we can again construct a coupling such that, if the supposedly slower process moves, the supposedly faster process does as well almost surely.
Indeed, either both processes are in the same state, then we find the stochastic dominance by the decision trees, or the fast process has more agents of opinion $a$. 
But as stochastic dominance is transitive, we can construct a coupling via the triangle inequality.

Finally, we will describe the overall coupling of the two processes as the \emph{path-coupling} along those couplings per step which will prove the first part of \cref{thm:main-result}. 
The second part will follow analogously as we can show via the decision trees that in the comparison of $P_{2j-1}$ and $P_{2j}$, the chance to obtain the same state in the next step is equal under both processes while in the comparison of $P_{2j}$ and $P_{2j+1}$ those decision trees show that the probability of increasing the majority opinion is larger in $P_{2j+1}$.

\begin{observation}
The processes $P_{1}$ and $P_3$ have, almost surely, a finite stopping time. 
\end{observation}
For the sequential process, we show this for $P_3$ in \cref{sec:3Majority}, while for $P_1$ this follows by the results of \Textcite{DBLP:conf/soda/SchoenebeckY18}.
For the Gossip Model, this is proven in \cite{DBLP:conf/podc/GhaffariL18}.
In this setting, Strassen's Theorem guarantees the existence of a coupling $\gamma \in \cP \bc{ \cE^2 }$ of $\mu$ and $\nu$ with the following property.
\begin{theorem}[Strassen's Theorem \cite{strassen1965}] \label{thm:strassen}
Let $\mu, \nu$ be probability measures on a Polish space endowed with a partial ordering $\preceq$ such that $\mu$ stochastically dominates $\nu$. Let $X \sim \mu$ and $Y \sim \nu$, then there is a coupling $\gamma$ of $\mu$ and $\nu$ such that, if $\bc{ \hat X, \hat Y } \sim \gamma$, we have
\begin{align*}
    X {\buildrel d \over =} \hat X, \qquad Y {\buildrel d \over =} \hat Y \qquad \text{and} \qquad \Prob{ \hat Y \preceq \hat X } = 1.
\end{align*}
\end{theorem}

\begin{restatable}{lemma}{LemDominance}
\label{Lem_dominance}
We find for $P_{2j}$ and $P_{2j+1}$ the following. Let $X_t$ denote the number of agents with majority opinion at time $t$. If $s > s'$, then for all $d \in \cbc{0, 1, ..., n}$
\begin{align*}
    \Prob{ X_{t+1} \geq d \mid X_{t} = s } \geq \Prob{ X_{t+1} \geq d \mid X_{t} = s' }.
\end{align*}
\end{restatable}
We provide the detailed calculation in \cref{proof_Lem_dominance} and get the following corollary.
\begin{corollary} \label{cor_dominance}
For any two processes $P, P'$, we find the following stochastic dominance. Let $X_t$ denote the number of agents with opinion $a$ with respect to process $P$ at time $t$ and let $X'_t$ be the analogous quantity with respect to $P'$.  
Assume that for any $d \in [n]$
\begin{align*}
    \Prob{ {X_{t+1}} \geq d \mid \bc{X_{t}} = s} \geq \Prob{ {X'_{t+1}} \geq d \mid {X'_{t}} = s},
\end{align*}
then we have also
\begin{align*}
    \Prob{ {X_{t+1}} \geq d \mid {X_{t}} = s + t'} \geq \Prob{ {X'_{t+1}} \geq d \mid {X'_{t}} = s},
\end{align*}
for any $d \in [n]$ and $t' > 0$ such that $s + t' \leq n$.
\end{corollary}
Let $X^{(k)}_{t}$  denote the number of agents with majority opinion after step $t$ of process $\Pn{k}$  for any $k \in \mathbb{N}$.
Furthermore, for a given agent $x$ we denote by $x_t^{(k)}$ its opinion in process $\Pl{k}$ at time $t$.
In the following we compare two processes with each other. The comparisons of $\Pl{2j}$ (even) to $\Pl{2j+1}$ (odd) and  $\Pl{2j-1}$ (odd) to $\Pl{2j}$ (even) require slightly different calculations.
Therefore, we have to show two similar lemmas for these two cases, \cref{lem:diffPoss2jvs2jplus1} for the former case and \cref{lem:diffPoss2jminus1vs2j} for the latter case.

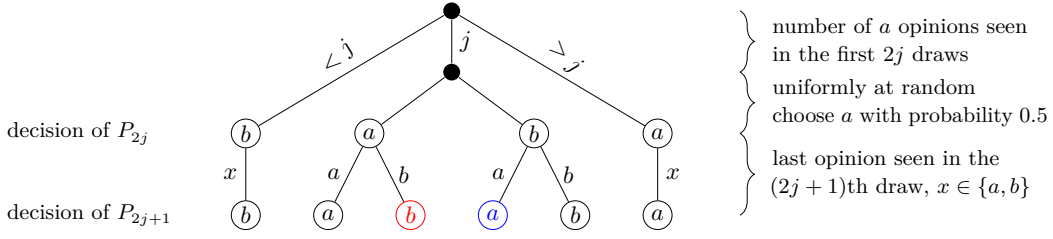
\begin{figure*}[t]
    \resizebox{\columnwidth}{!}{%
	\begin{tikzpicture}[%
	,scale=.6
	,dot/.style={circle, fill=black, inner sep=0pt, minimum size=7pt}
	,ring/.style={circle, draw, inner sep=0pt, minimum size=12pt}
	]
	\pgfmathsetmacro{\offsRight}{16}
	\pgfmathsetmacro{\zeroLvl}{10}
	\pgfmathsetmacro{\firstLvl}{8.5}
	\pgfmathsetmacro{\secondLvl}{7}
	\pgfmathsetmacro{\thirdLvl}{5}
	
	\node[dot]  (n1) at (9,\zeroLvl)        {};
	\node[ring] (n21) at (4,\secondLvl)     {$b$};
	\node[dot]  (n22) at (9,\firstLvl)      {};
	\node[ring] (n23) at (14,\secondLvl)    {$a$};
	
	\draw (n1) -- node[midway, above,rotate=45]{$< j$} (n21);
	\draw (n1) -- node[midway, right]{$j$}   (n22);
	\draw (n1) -- node[midway, above,rotate=-45]{$>j$}  (n23);
	
	\node[ring] (n31) at (7,\secondLvl)  {$a$};
	\node[ring] (n32) at (11,\secondLvl) {$b$};
	
	\draw (n22) -- node[midway, right]{} (n31);
	\draw (n22) -- node[midway, right]{} (n32);

	\node[ring] (n41) at (4,\thirdLvl) {$b$};
	\node[ring] (n42) at (6,\thirdLvl) {$a$};
	\node[ring, color = red] (n43) at (8,\thirdLvl) {$b$};
	\node[ring, color=blue] (n44) at (10,\thirdLvl) {$a$};
	\node[ring] (n45) at (12,\thirdLvl) {$b$};
	\node[ring] (n46) at (14,\thirdLvl) {$a$};
	
	\draw (n21) -- node[midway, left]{$x$}  (n41);
	\draw (n31) -- node[midway, left]{$a$}  (n42);
	\draw (n31) -- node[midway, right]{$b$} (n43);
	\draw (n32) -- node[midway, left]{$a$}  (n44);
	\draw (n32) -- node[midway, right]{$b$} (n45);
	\draw (n23) -- node[midway, right]{$x$} (n46);

	\draw[decorate,decoration={brace,amplitude=6pt}] (\offsRight,\zeroLvl) -- node [right=10pt,align=left] {\parbox{4cm}{\small number of $a$ opinions seen\newline in the first $2j$ draws}} (\offsRight,\firstLvl);
	\draw[decorate,decoration={brace,amplitude=6pt}] (\offsRight,\firstLvl) -- node [right=10pt,align=left] {\parbox{4cm}{\small uniformly at random\newline choose $a$ with probability $0.5$}} (\offsRight,\secondLvl);
	\draw[decorate,decoration={brace,amplitude=6pt}] (\offsRight,\secondLvl) -- node [right=10pt,align=left] {\parbox{4cm}{\small last opinion seen in the\newline \mbox{$(2j+1)$th} draw, $x \in \{a,b\}$ }} (\offsRight,\thirdLvl);
	
	\node[anchor=west] at (-2,\secondLvl) {\small decision of $\Pn{2j}$};
	\node[anchor=west] at (-2,\thirdLvl) {\small decision of $\Pn{2j+1}$};
	\end{tikzpicture}
}
	\caption{Decision tree comparing $\Pn{2j}$ and $\Pn{2j+1}$}
	\label{fig:2jvs2jplus1}
\end{figure*}
First we compare two successive processes $\Pl{2j}$ and $\Pl{2j+1}$.
The following lemma states that in process $\Pl{2j+1}$ it is more likely for an agent with opinion $b$ to change to $a$ while in process $\Pl{2j}$ it is more likely that an agent with opinion $a$ changes to opinion $b$ than in the other process respectively.
\begin{restatable}{lemma}{diffPossTwojvsTwojplusOne}
\label{lem:diffPoss2jvs2jplus1}
Let $x$ be an agent that is updated in the next step, $x_t^{(k)}$ its opinion in process $\Pl{k}$ at time $t$, $s \in [n]$ and $\alpha = \frac{s}{n}$. It holds that
\begin{align*}
    \MoveEqLeft \Prob{x_{t+1}^{(2j+1)} = a | x_t^{(2j+1)} =b,   X^{(2j+1)}_{t}= s}\\ 
    &= \Prob{x_{t+1}^{(2j)}= a | x_t^{(2j)} =b,  X^{(2j)}_{t}= s} + \frac{(2\alpha -1)}{2} \binom{2j}{j} \alpha^{j} (1-\alpha)^{j}
\intertext{and}
\MoveEqLeft\Prob{x_{t+1}^{(2j+1)} = b | x_t^{(2j+1)} =a,   X^{(2j+1)}_{t}= s} \\
&= \Prob{x_{t+1}^{(2j)} = b | x_t^{(2j)} =a,  X^{(2j)}_{t}= s}
- \frac{(2\alpha -1)}{2} \binom{2j}{j} \alpha^{j} (1-\alpha)^{j}.
\end{align*}
\end{restatable}
To prove these equations it is sufficient to study the cases where the two processes have a different outcome. The probability for these cases directly reflects the difference in probability for that specific outcome. These cases are highlighted in \cref{fig:2jvs2jplus1}. We provide the detailed calculation in \cref{proof_lem_diffPoss2jvs2jplus1}

This difference in probabilities allows us to prove that, given the same state, $\Pl{2j+1}$ stochastically dominates the process $\Pl{2j}$ in the next step:
\begin{restatable}{lemma}{differentProbabilities}
\label{lem:differentProbabilities}
For each $j \in \mathbb{N}_0$ and each $s \in \mathbb{N}$ with $s > n/2$ and any $d \in [n]$ it holds that
\begin{align*}
\Prob{X^{(2j+1)}_{t+1} \geq d | X^{(2j+1)}_{t} = s}    &\geq  \Prob{X^{(2j)}_{t+1} \geq d | X^{(2j)}_{t} = s}.
\end{align*}
\end{restatable}
Note that this inequality follows trivially for $d \leq s-1$ and $d > s+1$. To prove the property for the cases $d=s$ and $d=s+1$ we can directly use the properties from \cref{lem:diffPoss2jvs2jplus1}.

\begin{figure*}[t]
    \resizebox{\columnwidth}{!}{%
	\begin{tikzpicture}[%
	,scale=.6
	,dot/.style={circle, fill=black, inner sep=0pt, minimum size=7pt}
	,ring/.style={circle, draw, inner sep=0pt, minimum size=12pt}
	]
	
	\pgfmathsetmacro{\offsRight}{16}
	\pgfmathsetmacro{\firstLvl}{6.6}
	\pgfmathsetmacro{\secondLvl}{4.4}
	\pgfmathsetmacro{\thirdLvl}{2.2}
	\pgfmathsetmacro{\fourthLvl}{0}
	
	\node[dot] (n1) at (9,\firstLvl) {};
	\node[ring] (n21) at (3,\secondLvl) {$b$};
	\node[ring] (n22) at (7,\secondLvl) {$b$};
	\node[ring] (n23) at (11,\secondLvl) {$a$};
	\node[ring] (n24) at (15,\secondLvl) {$a$};
	
	\draw (n1) -- node[midway, above, rotate=15]{$< j-1$} (n21);
	\draw (n1) -- node[midway, below, rotate=45]{$j-1$} (n22);
	\draw (n1) -- node[midway, right]{$j$} (n23);
	\draw (n1) -- node[midway, above]{$>j$} (n24);
	
	\node[ring] (n31) at (3,\thirdLvl) {$b$};
	\node[ring] (n32) at (5,\thirdLvl) {$b$};
	\node[dot] (n33) at (7,\thirdLvl) {};
	\node[dot] (n34) at (11,\thirdLvl) {};
	\node[ring] (n35) at (13,\thirdLvl) {$a$};
	\node[ring] (n36) at (15,\thirdLvl) {$a$};
	
	\draw (n21) -- node[midway, right]{$x$} (n31);
	\draw (n22) -- node[midway, right]{$b$} (n32);
	\draw (n22) -- node[midway, right]{$a$} (n33);
	\draw (n23) -- node[midway, right]{$b$} (n34);
	\draw (n23) -- node[midway, right]{$a$} (n35);
	\draw (n24) -- node[midway, right]{$x$} (n36);
	
	\node[ring] (n41) at (6,\fourthLvl) {$b$};
	\node[ring,color=blue] (n42) at (8,\fourthLvl) {$a$};
	\node[ring, color = red] (n43) at (10,\fourthLvl) {$b$};
	\node[ring] (n44) at (12,\fourthLvl) {$a$};
	
	\draw (n33) -- node[midway, left]{} (n41);
	\draw (n33) -- node[midway, left]{} (n42);
	\draw (n34) -- node[midway, right]{} (n43);
	\draw (n34) -- node[midway, right]{} (n44);

	\draw[decorate,decoration={brace,amplitude=6pt}] (\offsRight,\firstLvl) -- node [right=12pt,align=left] {\parbox{4cm}{\small number of $a$ opinions seen\newline in the first $2j-1$ draws}} (\offsRight,\secondLvl);
	\draw[decorate,decoration={brace,amplitude=6pt}] (\offsRight,\secondLvl) -- node [right=12pt,align=left] {\parbox{4cm}{\small last opinion seen in the\newline $(2j)$th draw, $x \in \{a,b\}$}} (\offsRight,\thirdLvl);
	\draw[decorate,decoration={brace,amplitude=6pt}] (\offsRight,\thirdLvl) -- node [right=12pt,align=left] {\parbox{4cm}{\small uniformly at random\\ choose $a$ with probability $0.5$}} (\offsRight,\fourthLvl);
	
	\node[anchor=west] at (-2,\secondLvl) {\small decision of $\Pn{2j-1}$};
	\node[anchor=west] at (-2,\thirdLvl) {\small decision of $\Pn{2j}$};
	\node[anchor=west] at (-2,\fourthLvl) {\small decision of $\Pn{2j}$};
	\end{tikzpicture}
}	
	\caption{Decision tree comparing $\Pn{2j-1}$ and $\Pn{2j}$}
	\label{fig:2jplus1vs2jplus2}
\end{figure*}
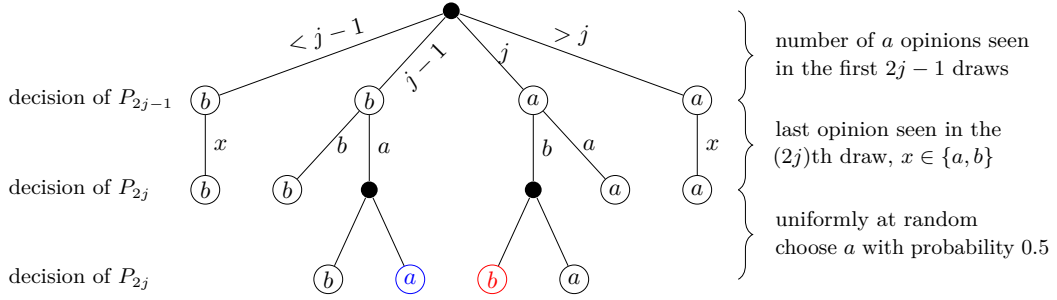

On the other hand, when comparing the processes $\Pl{2j-1}$ and $\Pl{2j}$ with respect to the difference in probability for an agent to change its opinion, we note that there is no difference in the probabilities given that all agents are in the same state.
\begin{restatable}{lemma}{diffPossTwojminusOnevsTwoj}
\label{lem:diffPoss2jminus1vs2j}
Let $x$ be an agent that is updated in the next step, $x_t^{(k)}$ its opinion in process $\Pl{k}$ at time $t$, and $s \in [n]$. It holds that
\begin{align*}
\Prob{x_{t+1}^{(2j-1)} = a | x_t^{(2j-1)} =b,   X^{(2j-1)}_{t}= s} 
&= \Prob{x_{t+1}^{(2j)} = a | x_t^{(2j)} =b,  X^{(2j)}_{t}= s}
\intertext{and}
\Prob{x_{t+1}^{(2j-1)} = b | x_t^{(2j-1)} =a,   X^{(2j-1)}_{t}= s} &= \Prob{x_{t+1}^{(2j)} = b | x_t^{(2j)} =a,  X^{(2j)}_{t}= s}.
\end{align*}
\end{restatable}
A similar statement has previously been shown by \textcite{DBLP:journals/dc/FraigniaudN19} for a related model.
The proof of \cref{lem:diffPoss2jminus1vs2j} is analogous to the proof of \cref{lem:diffPoss2jvs2jplus1}.
For completeness, it can be found in \cref{proof_lem_diffPoss2jminus1vs2j}.

\begin{restatable}{lemma}{differentProbabilitiesTwo}
\label{lem:differentProbabilities2}
For each $j \in \mathbb{N}_0$ and each $s \in \mathbb{N}$ with $s > n/2$ and any $d \in [n]$ it holds that
\begin{align*}
\Prob{X^{(2j)}_{t+1} \geq d | X^{(2j)}_{t} = s}    & = \Prob{X^{(2j-1)}_{t+1} \geq d | X^{(2j-1)}_{t} = s}.
\end{align*}
\end{restatable}
The proof can be found in \cref{proof_lem_differentProbabilities}.
We are now ready to put everything together and prove our main result for the sequential model.
\begin{proof}[Proof of \cref{thm:main-result}]
We prove \cref{thm:main-result} by induction given the initial state $X_0$ and start with the case $  T_{2j} \bc{ X_0 } \preceq   T_{2j+1} \bc{ X_0 }$.
Given $X_0$, \cref{lem:differentProbabilities} guarantees that for all $s > 0$
\begin{align*}
    \Prob{ X_1^{(2j + 1)} \geq s \mid X_0 } \geq  \Prob{ X_1^{(2j)} \geq s \mid X_0 }.
\end{align*}
Therefore, by \cref{thm:strassen}, we find a coupling $\gamma_{1}$ such that under $\gamma$, $X_1^{(2j + 1)} \geq X_1^{(2j)}$ almost surely.
Now, assume that we constructed a coupling $\gamma^{(t)} = \gamma_1 \otimes \ldots \otimes \gamma^{(t)}$ of $\bc{ X_1^{2j}, \ldots, X_{t}^{2j} }$ and $\bc{ X_1^{2j+1}, \ldots, X_{t}^{2j+1} }$. Under $\gamma^{(t)}$ we have by induction hypothesis that 
\begin{align*}
     \Pr_{\gamma^{(t)}} \bc{ X_{t}^{2j+1} \geq X_{t}^{2j} } = 1.
\end{align*}
Therefore, by \cref{cor_dominance} and \cref{lem:differentProbabilities}, we find given $X_t^{(2j + 1)} \geq X_t^{(2j)}$ that
\begin{align*}
    \Prob{ X_{t+1}^{(2j + 1)} \geq s \mid X_t^{(2j + 1)} } \geq  \Prob{ X_{t+1}^{(2j)} \geq s \mid X_t^{(2j)} }.
\end{align*}
Thus, \cref{thm:strassen} implies, given $X_t^{(2j + 1)} \geq X_t^{(2j)}$ the existence of a coupling $\gamma_{t+1}$ such that
\begin{align*}
     \Pr_{\gamma_{t+1}} \bc{ X_{t+1}^{2j+1} \geq X_{t+1}^{2j} } = 1.
\end{align*}
We define
    $\gamma^{(t+1)} = \gamma^{(t)} \otimes \gamma_{t+1}$
and  $  T_{2j} \bc{ X_0 } \preceq   T_{2j+1} \bc{ X_0 }$ follows by induction.

Next, we need to prove that $  T_{2j-1} \bc{ X_0 } \preceq   T_{2j} \bc{ X_0 }$. This follows completely analogously with \cref{lem:differentProbabilities} replaced by \cref{lem:differentProbabilities2}.

Finally, we need to construct the bounds on the expectation. Given the coupling $\gamma^{  T(2j+1)(X_0)}$ of $P_{2j+1}$ and $P_{2j}$, we find that, under this coupling, for every step $t = 1 \ldots   T_{2j+1}$, we have $X_{t}^{(2j+1)} \geq X_{t}^{(2j)}$ almost surely and therefore $\Erw \brk{   T_{2j+1}(X_0)} \leq \Erw \brk{   T_{2j}(X_0)}$. The second expectation is a bit more subtle. While it is analogously easy to prove that $\leq$ holds due to the constructed coupling, the equality in expectation needs to be conducted explicitly.
To this end, we get from \cref{lem:differentProbabilities2} that 
\begin{align*}
    \Prob{ X_{t}^{(2j-1)} < n \mid X_{t-1} = s} = \Prob{ X_{t}^{(2j)} < n \mid X_{t-1} = s}.
\end{align*}
Therefore, inductively,
\begin{align*}
    \Prob{ X_{t}^{(2j-1)} < n \mid X_{0} = s} = \Prob{ X_{t}^{(2j)} < n \mid X_{0} = s}.
\end{align*}
But then,
\begin{align*}
    \Erw \brk{   T_{2j - 1} \bc{X_0} } & = \sum_{t \geq 0} \Prob{   T_{2j - 1} \bc{X_0} > t } = \sum_{t \geq 0} \Prob{ X_{t}^{(2j-1)} < n \mid X_0 } \\
    & = \sum_{t \geq 0} \Prob{ X_{t}^{(2j)} < n \mid X_0 } = \sum_{t \geq 0} \Prob{   T_{2j} \bc{X_0} > t } = \Erw \brk{   T_{2j} \bc{X_0} }. \qedhere
\end{align*}
\end{proof}

\subsection{Gossip Model}\label{sec:gossip}
We now extend the previous analysis to the gossip model. Recall that in this model all agents are activated in parallel rounds.
In such a round, all agents sample $j$ other agents $v_1, \dots, v_j$ u.a.r.
Then they compute their new opinion as the majority opinion among the sample, breaking ties u.a.r.
Here, the agents use the opinions of the other agents from the beginning of the round.
At the end of the round (once all agents have computed the new opinion) all agents synchronously update their opinion to the new value.

\begin{proof}[Proof of \cref{thm:main-result} for the Gossip Model]
In our extended analysis we use a coupling of the two parallel processes similarly to the coupling of one step of the gossip model.
Observe that in process $P_{2j}$ every agent samples $2j$ agents u.a.r., while in process $P_{2j+1}$ every agent samples $2j+1$ agents. 
Therefore, process $P_{2j}$ makes $2j \cdot n$ random choices from $[n]$ in each round, while $P_{2j+1}$ makes $(2j+1)\cdot n$ random choices.
We use the straight-forward coupling and define that the $2j$ choices of every agent $u$ in $P_{2j}$ are identical to the first $2j$ choices of agent $u$ in process $P_{2j+1}$.

We now analyze the deviation of the two processes that stems from the $2j+1$th additional choice in process $P_{2j+1}$. 
Here we observe the following.
In each round of process $P_{2j}$ there are three disjoint sets of agents, $M_a, M_b,$ and $M_u$.
The sets $M_a$ and $M_b$ are comprised of agents that sample at least $j+1$ agents of the majority opinion $a$ and the minority opinion $b$, respectively.
All other agents are in $M_u$.
The agents in $M_a$ will adopt opinion $a$ at the end of the round in both processes:
the $j+1$ samples of opinion $a$ is larger than the winning margin in both processes, which is $j$ in $P_{2j}$ and $(2j+1)/2$ in $P_{2j+1}$.
Analogously, the agents in $M_b$ will adopt opinion $b$ in both processes.
Finally, the interesting group are the $M_u$ agents. 
These agents have sampled a tie in process $P_{2j}$, meaning they have sampled $j$ agents with opinion $a$ and another $j$ agents with opinion $b$.
This means, in process $P_{2j}$ all agents in $M_u$ adopt either opinion $a$ or opinion $b$ with probability $1/2$ each.
In process $P_{2j+1}$, however, the $2j+1$th sample makes the decision.
(Recall that in a process $P_{2j+1}$ with an odd number of samples and $k=2$ opinions no ties are possible.)
Therefore, in process $P_{2j+1}$ all agents in $M_u$ adopt opinion $a$ with probability $\alpha$ and opinion $b$ with probability $(1-\alpha)$.

Summarizing, we have the following.
Due to the coupling of $P_{2j}$ with $P_{2j+1}$, all agents in $M_a$ or $M_b$ behave exactly the same in both processes.
We use $Z_a = |M_a|$ and $Z_b = |M_b|$ to denote their respective numbers.(Observe that $Z_a$ and $Z_b$ are the same in $P_{2j}$ and $P_{2j+1}$ due to the coupling.)
In the following, we condition on the event that $|M_u| = m_u$.
For the agents in $M_u$, the outcome can be described by binomial random variables:
let $Z_u^{2j}$ in process $P_{2j}$ and $Z_u^{2j+1}$ in process $P_{2j+1}$ be the numbers of agents in $M_u$ that adopt opinion $a$.
Then 
\begin{align*}
Z_u^{2j} \sim \Bin(m_u,1/2) && \text{ and } && Z_u^{2j+1} \sim \Bin(m_u,\alpha)
\end{align*}
with $\alpha \geq 1/2$.
Irrespective of the value of $m_u$ we observe from well-known properties of binomial distributions 
that $Z_u^{2j}$ is stochastically dominated by $Z_u^{2j+1}$, and hence
\[
X_{t+1}^{(2j)} = Z_a + Z_u^{2j} \prec Z_a + Z_u^{2j+1} = X_{t+1}^{(2j+1)} .
\]

The proof for the dominance of $P_{2j}$ over $P_{2j-1}$ uses similar definitions and follows analogously, with exception that $M_u$ represents the nodes that are undecided after the first $2j-2$ draws and that $Z_u^{2j-1}$ and $Z_u^{2j}$ follow the same binomial distribution $\Bin(m_u,\alpha)$.

The only ingredient that is left to prove is the monotonicity within one specific process. Indeed, if an analogous result as \cref{Lem_dominance} in the sequential model can be proven, the path coupling argument follows the same lines as in the previous section.
\begin{restatable}{lemma}{LemDominanceGossip}
\label{Lem_dominanceGossip}
We find for $P_{2j}$ and $P_{2j+1}$ in the Gossip model following. Let $X_t$ denote the number of agents with majority opinion at time $t$. If $s > s'$, then for all $d \in \cbc{0, 1, ..., n}$
\begin{align*}
    \Prob{ X_{t+1} \geq d \mid X_{t} = s } \geq \Prob{ X_{t+1} \geq d \mid X_{t} = s' }.
\end{align*}
\end{restatable}
\begin{proof}
As before, let $M_a$ and $M_b$ denote the sets of agents that sample at least $j+1$ agents of the majority opinion $a$ and the minority opinion $b$, respectively.
If $s > s'$, the monotonicity of the binomial distribution yields
\begin{align*}
    \abs{M_a}_{\mid X_t = s} \succeq \abs{M_a}_{\mid X_t = s'} \text{ and } \abs{M_b}_{\mid X_t = s} \preceq \abs{M_b}_{\mid X_t = s'}. 
\end{align*}
Therefore, the lemma follows from Strassen's theorem.
\end{proof}
Now the path coupling follows analogously to the previous section.
\end{proof}

\subsection{Analysis of {3}-Majority}
\label{sec:3Majority}

In this section we analyze $3$-Majority in the sequential model.
We start with an overview of the proof of \cref{thm:convergencetime}.
The proof consists of three parts.
The first part follows along the lines of the proof by \textcite{DBLP:journals/nc/CondonHKM20} for the related approximate majority process in tri-molecular chemical reaction networks.
It shows that we preserve the initial majority (assuming a bias of $\sqrt{n \log n}$) and reach a bias of $\epsilon n$ within $O(n \log n)$ time \whp.
(Recall that the \emph{bias} is defined as the difference of the numbers of agents supporting opinion $a$ and opinion $b$.)
The proof is based on the following result for gambler's ruin from \cite{Feller1968}.
\begin{lemma}[Asymmetric one-dimensional random walk, {\cite[\nopp XIV.2]{Feller1968}}, version from \cite{DBLP:journals/nc/CondonHKM20}] \label{lem:feller}
If we run an arbitrarily long sequence of independent trials, each with success probability at least $p$, then the probability that the number of failures ever exceeds the number of successes by $b$ is at most $\bc{\frac{1-p}{p}}^b$.
\end{lemma}

In the second part we use a drift analysis based on \cite{Lengler2020} to show that we reach consensus on the initial majority opinion quickly once we have a bias of order $\Omega{n}$.
The proof is based on a carefully conducted drift-analysis, where we use the following fairly recent result.
\begin{theorem}[Special case of Theorem 18 of \cite{Lengler2020}] \label{thm:lengler}
Let $\cbc{Y_t}_{t \geq 0}$ be a sequence of non-negative random variables with a finite state space $S \subset \mathbb{R}_{\geq 0}$ such that $0 \in S$. Define 
\begin{align*}
    s_{\min} = \min \bc{S \setminus \cbc{0}} \quad \text{and} \quad   T = \inf \cbc{ t \geq 0 \mid Y_0 = 0 }.
\end{align*}
If $Y_0 = s_0$ and there is $\delta > 0$ (independent from $t$) such that for all $s \in S \setminus \cbc{0}$ and all $t \geq 0$ we have
\begin{align*}
    \Erw \brk{ Y_t - Y_{t+1} \mid Y_t = s } \geq \delta s,
\end{align*}
then, for all $r \geq 0$,
\begin{align*}
     \Prob{   T > \ceil*{ \frac{r + \log \bc{ s_0 / s_{\min} }   }{\delta} } } \leq e^{-r}.
\end{align*}
\end{theorem}

In the third part we again show that the analysis from \cite{DBLP:journals/nc/CondonHKM20} is applicable in our setting if we do not have an initial bias.
All three parts together prove the first statement of our theorem.
The second statement follows from part one together with part two.

\paragraph{Part 1} We start with the first part.
We follow along the lines of  \cite{DBLP:journals/nc/CondonHKM20} and use \cref{lem:feller} to show the following statement.

\begin{lemma} \label{lem:condon}
Let $\Delta_t$ be the additive bias at time $t$.
With probability $1 - e^{-\Omega(\Delta_t^2/n)}$, the bias $\Delta_t$ does not drop below $\Delta_t/2$ and increases to $\min\set{2\Delta_t, n}$ within $2n$ time steps.
\end{lemma}
\begin{proof}
Let $X_t$ denote the number of agents with the majority opinion at time $t$ and let $Y_t = n - X_t$ denote the number of agents with the minority opinion at time $t$.
We analyze our process as a variant of gamblers' ruin and apply \cref{lem:feller}.
We only consider \emph{productive} steps in which the number of agents of a specific opinion changes.
For $X_t \in \bc{ \frac{n}{2}, \frac{n}{2} + \eps n}$ it holds that $\Prob{ X_{t+1} \neq X_t } = \Omega(1)$ and hence conditioning on productive steps only increases the constants hidden in the asymptotic notation.

In each productive step, the success probability reads $p=\Prob{X_{t+1}>X_{t} \mid X_{t+1} \neq X_t}$ and the failure probability reads $1-p=\Prob{X_{t+1} < X_{t} \mid X_{t+1} \neq X_t}$.
Let $\Delta_t = X_t - \frac{n}{2}$ denote the bias at time $t$.
We have for any $\Delta$ that
\begin{align} \label{eq:probabilities-bias}
\frac{1-p}{p}=\frac{2\bc{\frac{1}{2}-\frac{\Delta}{n}}-3\bc{\frac{1}{2}-\frac{\Delta}{n}}^3+\bc{\frac{1}{2}-\frac{\Delta}{n}}}{2\bc{\frac{1}{2}-\frac{\Delta}{n}}^4-5\bc{\frac{1}{2}-\frac{\Delta}{n}}^3+3\bc{\frac{1}{2}-\frac{\Delta}{n}}^2}< 1-16\frac{\Delta}{n}.
\end{align}
Unfortunately, the success probabilities vary over time as they depend on the bias. We proceed to bound the probabilities from below.

Let $\Delta_0$ be the bias at time $t = 0$ and let $\cR$ denote the following event: during $2n$ productive steps we always have at least half of the initial bias, i.e., $\cR = \cbc{\forall 1 \leq i \leq 2n\colon \Delta_{i} \geq  \Delta_0/2  }$.
From \cref{lem:feller} we get with $b= \Delta_0/2$ that \begin{equation} \label{eq:prob-R}
\Prob{ \cR } \geq 1 - e^{-\Omega(\Delta_0^2/n)}. \end{equation}

Similarly to \cite{DBLP:journals/nc/CondonHKM20}, we couple the productive steps of the $3$-Majority process with a biased random walk with (fixed) success probability $p>\frac{1}{2}+\frac{\Delta_0}{4n}$.
As \eqref{eq:probabilities-bias} is monotonously decreasing in $\Delta$, the number of steps required by the biased random walk to increase the bias stochastically dominates the number of steps that $3$-Majority requires.
It follows from Chernoff bounds that the random walk reaches $2\Delta_0$ within $2n$ time steps with probability $1 - e^{-\Omega(\Delta_t^2/n)}$.
Together with \eqref{eq:prob-R} the statement follows.
\end{proof}

We now use \cref{lem:condon} and show that if there is a small bias of size $\sqrt{n \log n}$ then within $O(n \log n)$ rounds there will be a bias of size $\Omega{(n)}$ \whp.
\begin{corollary}\label{cor:condon}
Assume $X_0 = \frac{n}{2} + \sqrt{n \log n}$. 
Then there is a time $t = O(n \log n)$ such that $X_t > \frac{1+\eps}{2}n$ for some constant $\eps > 0$ \whp.
Moreover, the initial majority opinion is preserved.
\end{corollary}
\begin{proof}
The proof follows by applying \cref{lem:condon} $O(\log n)$ times.
We remark that the initial majority opinion is preserved since the random walk modeling the bias never returns to zero.
\end{proof}

\paragraph{Part 2} We now show the second part, where we prove that the process converges within $O(n \log n)$ further steps once we have a bias of $\eps n$.
Let $Y_t$ denote the number of agents of the \emph{minority} opinion at time $t$ and assume that $Y_0 \leq \frac{n}{2} - \eps n$. In a first step, we claim that the process will not improve the minority opinion severely if only $C n \log n$ steps are conducted for some large constant~$C$.
\begin{lemma}\label{lem:drift}
Assume $Y_0 \leq \frac{n}{2} - \eps n$.
Then there is a time $t = O(n \log n)$ such that $Y_t = 0$ \whp.
Moreover, $Y_{t'} \leq \frac{1 - \eps}{2}n$ for all $t' \leq t$.
\end{lemma}

\begin{proof}
We start the proof by showing the following claim:

\begin{claim*}$Y_{t'} \leq \frac{1 - \eps}{2}n$ for all $t' = O( n \log n )$ \whp.\end{claim*}

\noindent This is an immediate consequence of the following coupling.
Let $R_t$ be the (unbiased) random walk on $\mathbb{Z}$.
It is a well known fact that after $T$ steps the random walk $R_t$ has distance at most $O(\log^2 n \cdot \sqrt{n})$ from the origin \whp.
By construction, $R_t \preceq Y_t$ and the claim follows.

\medskip

We now calculate $\Erw \brk{ Y_t - Y_{t+1} \mid Y_t = s }$ for $P_3$ in the sequential model. Given $Y_t = s$, let $p_s(a,b)$ be the probability to increase the minority opinion by one and let $p_s(b,a)$ be the probability to decrease the minority opinion by one. Then,
\begin{align*}
   \Erw \brk{ Y_t - Y_{t+1} \mid Y_t = s } = p_s(b,a) - p_s(a,b). 
\end{align*}
We observe
\begin{align*}
    p_s(a,b) & = \frac{n-s}{n} \Prob{ \Bin \bc{ 3, \frac{s}{n} } \geq 2 }, & & &
    p_s(b,a) & = \frac{s}{n} \Prob{ \Bin \bc{ 3, \frac{s}{n} } \leq 1},
\end{align*}
and therefore,
\begin{align*}
  \frac{p_s(b,a) - p_s(a,b)}{b} = \frac{ 2 s^2 - 3 s n + n^2}{n^{3}} .
\end{align*}
We define $\delta_s = \frac{p_s(b,a) - p_s(a,b)}{s}$ and observe 
\begin{align*}
 \delta_{s} - \delta_{s-1} =  \frac{4s - 3n - 2}{n^3} < 0 \quad \text{ if } \quad s \leq 0.75 n. 
\end{align*}
Therefore, since $s \leq \frac{1 - \eps}{2}n$ by the previous claim, $\delta_s$ is monotonously decreasing in $s$. Furthermore, 
\begin{align*}
  \delta_{\frac{1 - \eps}{2}n} = \frac{(1 + \eps/2)\eps}{4n} \quad \text{ and } \quad \delta_{1} = n^{-1} + O \bc{n^{-2}}.
\end{align*}
 Thus, we apply \cref{thm:lengler} with 
\begin{align*}
 \delta = \frac{(1 + \eps/2)\eps}{4n}, \quad s_0 = \bc{\frac{1}{2} - \eps}n, \quad r = \log n, \quad \text{and} \quad s_{\min} = 1
\end{align*}
and the statement follows.
\end{proof}

\paragraph{Part 3}
It remains to show the third part of the proof. We observe the following.
We use the same \emph{checkpoint states} $g_j$ as in \cite{DBLP:journals/nc/CondonHKM20} where $g_0 = 0$ and $g_j =  2^{j+3} \cdot \sqrt{n}$. A checkpoint state can be intuitively described as follows. We let $P_3$ run in packages of $2n$ productive update steps and monitor the majority opinion. Suppose we are in checkpoint state $g_1 = 8 \sqrt{n}$. After $2n$ productive updates, \cref{lem:condon} guarantees that with probability at least $1 - 1/(2^j + O(1))$ the majority opinion exceeds $g_2$. Now we interpret this process as a (biased) random walk on the checkpoint states $\cbc{g_j}_j$ in which every conducted step consists of $2n$ productive update steps of $3-$Majority.   
Analogously to the analysis of  \cite{DBLP:journals/nc/CondonHKM20}, it holds that
\begin{enumerate}
\item the transition between checkpoint states $g_0$ and $g_1$ has probability $\Omega(1)$, and 
\item for $j \geq 1$ the transition between checkpoint states $g_j$ and $g_j+1$ has
probability at least $1 - 1/(2^j + O(1))$.
\end{enumerate}
As in \cite{DBLP:journals/nc/CondonHKM20}, the first statement follows from a coupling with an unbiased random walk, and the second statement follows from \cref{lem:condon}.
It follows from the analysis in \cite[Section 3.2]{DBLP:journals/nc/CondonHKM20} that $3$-Majority reaches a bias of $\sqrt{n \log n}$ within $O(n \log n)$ time. This proof is based on a careful trade-off between the geometrically increasing success probability $1 - 1/(2^j + O(1))$ to get into the next checkpoint state and the number of trials that are necessary to indeed reach the next state instead of falling back.

\medskip

\begin{figure*}[t]
\centering
\def\jM#1{{\small$#1$-Maj.}\hspace{-0.25em}}

\resizebox{\columnwidth}{!}{%
\fbox{
\input{figures/key}
}
}
\resizebox{\columnwidth}{!}{%
\input{figures/plot1-gossip}
\hfill
\input{figures/plot1-population}
}
\caption{
Average convergence time of $j$-Majority without initial bias and $j = 3, \dots, 12$ normalized over $\log n$ (gossip model) or $n \log n$ (sequential model).
Each data point shows the average of $100$ independent runs.
The left plot shows the gossip model and the right plot shows the sequential model.}
\label{fig:plot-1}

\resizebox{\columnwidth}{!}{
\input{figures/plot2-gossip}
\hfill
\input{figures/plot2-population}
}
\caption{Boxplots for the normalized convergence time of $j$-Majority without initial bias. The plots show details of the distribution of the same data as in \cref{fig:plot-1} for $n = 10^6$.}
\label{fig:plot-2}

\end{figure*}

With all three parts, we are now ready to put everything together and prove \cref{thm:convergencetime}.

\begin{proof}[Proof of \cref{thm:convergencetime}]
Assume there is no bias.
From the analysis in \cite{DBLP:journals/nc/CondonHKM20} we obtain (see above) that we reach a bias of size $\sqrt{n \log n}$ within $O(n \log n)$ time \whp.
From \cref{cor:condon} we obtain that within further $O(n \log n)$ time the bias is amplified to $\epsilon n$ for some constant $\epsilon > 0$ \whp.
Finally, from the drift analysis in \cref{lem:drift} we get that we converge in further $O(n \log n)$ time once we have a constant-factor bias \whp.
Together, this shows the first part of the theorem.

The second part of the theorem follows from \cref{lem:condon} and \cref{lem:drift}, where we observe that the initial majority opinion is preserved \whp. This concludes the proof.
\end{proof}

\section{Empirical Analysis}
\def\CC{C\nolinebreak[4]\hspace{-.05em}\raisebox{.4ex}{\relsize{-2}{\textbf{++}}}}

In this section we present simulation results to support our theoretical findings.
Our simulation software is implemented in the \CC{} programming language.
As a source of randomness it uses the Mersenne Twister \texttt{mt19937\_64} provided by the \CC{11} \texttt{\textless random\textgreater{}} library.
Our simulations have been carried out on two machines with two Intel(R) Xeon(R) E5-2630 v4 CPUs and 128 GiB of memory each running the Linux 5.13 kernel.
The simulation software and all required tools to reproduce our plots will be made publicly available upon publication of this paper.

In \cref{fig:plot-1} we plot the required number of rounds until $j$-Majority converges when each opinion is initially supported $n/2$ agents.
The data show the average convergence time over $100$ independent simulation runs for $j = 3, \dots, 12$. The number of agents $n$ is shown on the $x$-axis, and the normalized convergence time is shown on the $y$-axis. The left plot shows the data for the gossip model, where the normalization means that the required number of rounds is divided by $\log n$. The right plot shows the data for the sequential model, where the normalization means that the required number of interactions is divided by $n \log n$.

Our empirical data confirm our theoretical findings.
In particular, we observe that the processes exhibit a running time of $\Theta(\log n)$ rounds (gossip model) or $\Theta(n \log n)$ interactions (sequential model) for the values of $j$ we consider.
Furthermore, we clearly see that $\Ex{  T_{2j+2}(X_{0}) } = \Ex{  T_{2j+1}(X_{0}) }$ (i.e., 3-Majority converges as quickly as 4-Majority, 5-Majority converges as quickly as 6-Majority, and so on) and $\Ex{  T_{2j+1}(X_{0}) } \leq \Ex{  T_{2j}(X_{0}) }$ (i.e., 5-Majority is faster than 4-Majority, 7-Majority is faster than 6-Majority, and so on). This empirically confirms our results from \cref{thm:main-result} for both models, and it shows that the known results from the gossip model for $3$-Majority \cite{DBLP:conf/podc/GhaffariL18} carry over to the sequential model as predicted in \cref{thm:convergencetime}.

In the left plot in \cref{fig:plot-1} for the gossip model we additionally observe that the required number of rounds to reach consensus is slightly larger for smaller values of $n$.
This appears to be a consequence of the discrete rounds in the synchronous model: the observed deviation scales as $O(1/\log n)$, which is of the same size as the rounding error that arises when reporting the running time in discrete rounds of $n$ interactions each.

Finally, in \cref{fig:plot-2} we show additional detail for the distribution of the convergence times of the $j$-Majority processes with $n = 10^6$ and $j = 3, \dots, 12$.
Our boxplots show that the running times are strongly concentrated around the mean, and the constants hidden in the asymptotic analysis are small: the running time is less than $3 \log n$ rounds in the gossip model and less than $3 n \log n$ interactions in the sequential model.
The small constants hint at the practical applicability of the simple $3$-Majority process.

\section{Conclusions and Open Problems}

We analyze the family of $j$-Majority processes in two communication models with parallel and sequential activations.
In both models our results affirmatively answer an open question from \cite{DBLP:conf/podc/BerenbrinkCEKMN17} for the case of two opinions and prove the existence of a \emph{hierarchy}: our results show the stochastic dominance of the convergence time of the $j+1$-Majority process over the $j$-Majority process.
For $3$-Majority in the sequential model we show an asymptotically optimal bound of $O(n \log n)$ sequential activations. This matches the well-known bounds for the corresponding process in the gossip model.

An open question is whether a similar hierarchy exists for \emph{lazy} processes where agents keep their previous opinion if there is a tie among the sampled opinions.
A coupling between $3$-Majority and the (lazy) \protocol{TwoChoices} process was analyzed in \cite{DBLP:conf/podc/BerenbrinkCEKMN17}.
However, their general framework cannot be adapted to lazy processes for larger value of $j$: their analysis requires so-called \emph{AC-Processes} in which the next state of an agent depends only on the global opinion distribution  but not on the agent's current state. This is obviously not the case for lazy processes.
Note that our analysis also cannot be applied to lazy processes directly: \cref{lem:differentProbabilities,lem:differentProbabilities2} do not hold for lazy processes.

Another interesting open question considers the \emph{communication complexity} of a protocol instead which counts the number of interactions. Note that in $j$-Majority each activated agent interacts with $j$ agents. 
It would be interesting to rigorously analyze the trade-off between the convergence time and the communication complexity.

Finally, the most interesting open question is whether similar results can be shown for more than two opinions.
Unfortunately, our majoritzation-based approach does not generalize to $k > 2$.
The main reason is that natural monotonicity properties do not hold: the probability to increase the majority opinion does not only depend on the size of the majority opinion itself but instead on the entire opinion distribution.
This aligns well with a conjecture from \cite{DBLP:conf/podc/BerenbrinkCEKMN17} that states that counterexamples exist for any majorization attempt that uses a total order on opinion state vectors.
We believe that in order to show a hierarchy of majority protocols for more than two opinions different techniques will be needed.

\clearpage
\section*{References}
\printbibliography[heading=none]

\clearpage

\appendix

\section{Appendix} \label{apx:appendix}
\allowdisplaybreaks
\subsection{Proof of \texorpdfstring{\cref{Lem_dominance}}{Lemma 5}}\label{proof_Lem_dominance}
\LemDominance*
\begin{proof}
If $s \geq s'+2$, this is immediate. Indeed, we find 
\begin{align*}
    \Prob{ X_{t+1} \geq s-1 \mid X_{t} = s } = 1
    \text{ \ and \ } \Prob{ X_{t+1} \leq s-1 \mid X_{t} = s' } = 1
\end{align*} 
in this case. If, on the other hand, $s = s' + 1$, it is a sufficient condition that
\begin{align*}
    \MoveEqLeft \Prob{ X_{t+1} = s + 1 \mid X_{t} = s } + \Prob{ X_{t+1} = s \mid X_{t} = s } \geq \Prob{ X_{t+1} = s \mid X_{t} = s-1 } 
\end{align*}
or, equivalently,
\begin{align*}
    1 - \Prob{ X_{t+1} = s - 1 \mid X_{t} = s } \geq \Prob{ X_{t+1} = s \mid X_{t} = s-1 }. 
\end{align*}
We start with $P_{2j + 1}$. Here, we find
\begin{align*}
    1 - \Prob{ X_{t+1} = s - 1 \mid X_{t} = s } = 1 - \frac{s}{n} \Prob{ \Bin \bc{ 2j+1, \frac{s}{n} } \leq j }
\end{align*}
and
\begin{align*}
    \MoveEqLeft \Prob{ X_{t+1} = s \mid X_{t} = s-1 } \\
    &= \bc{1 - \frac{s-1}{n}}  \Prob{ \Bin \bc{ 2j+1, \frac{s-1}{n} } \geq j+1 } \\
    & = \bc{1 - \frac{s-1}{n}} \bc{ 1 - \Prob{ \Bin \bc{ 2j+1, \frac{s-1}{n} } \leq j }} \\
    & = 1 - \frac{s-1}{n} - \Prob{ \Bin \bc{ 2j+1, \frac{s-1}{n} } \leq j } + \frac{s-1}{n} \Prob{ \Bin \bc{ 2j+1, \frac{s-1}{n} } \leq j }.
\end{align*}
It therefore suffices to prove that
\begin{align} \label{eq_to_prove_2j}
    \MoveEqLeft s \Prob{ \Bin \bc{ 2j+1, \frac{s}{n} } \leq j } \\
    & \leq (s-1) + n \Prob{ \Bin \bc{ 2j+1, \frac{s-1}{n} } \leq j }  - (s-1) \Prob{ \Bin \bc{ 2j+1, \frac{s-1}{n} } \leq j } \notag
\end{align}
We have
\begin{align*}
 s \Prob{ \Bin \bc{ 2j+1, \frac{s}{n} } \leq j } \leq n \Prob{ \Bin \bc{ 2j+1, \frac{s-1}{n} } \leq j }
\end{align*}
since $s \leq n$ and because of the monotonicity of the Binomial distribution with respect to the success rate.
Furthermore,
\begin{align*}
 \MoveEqLeft (s-1) + n \Prob{ \Bin \bc{ 2j+1, \frac{s-1}{n} } \leq j } - (s-1) \Prob{ \Bin \bc{ 2j+1, \frac{s-1}{n} } \leq j } \\
 & \geq n \Prob{ \Bin \bc{ 2j+1, \frac{s-1}{n} } \leq j }.
\end{align*}
This verifies \eqref{eq_to_prove_2j}.
The calculus for $P_{2j}$ is similar. Here, we find
\begin{align*}
    \MoveEqLeft 1 - \Prob{ X_{t+1} = s - 1 \mid X_{t} = s } \\
    &= 1 - \frac{s}{n} \bc{ \Prob{ \Bin \bc{ 2j, \frac{s}{n} } \leq j-1 } + \frac{1}{2} \Prob{ \Bin \bc{ 2j, \frac{s}{n} } = j }}.
\end{align*}
Moreover,
\begin{align*}
    \MoveEqLeft \Prob{ X_{t+1} = s \mid X_{t} = s-1 } \\
    &= \bc{1 - \frac{s-1}{n}} \Bigg(\Prob{ \Bin \bc{ 2j, \frac{s-1}{n} } \geq j+1 }  + \frac{1}{2} \Prob{ \Bin \bc{ 2j, \frac{s-1}{n} } = j } \Bigg) \\
    &= \bc{1 - \frac{s-1}{n}} \Bigg( 1 - \Prob{ \Bin \bc{ 2j, \frac{s-1}{n} } \leq j-1 } - \frac{1}{2} \Prob{ \Bin \bc{ 2j, \frac{s-1}{n} } = j } \Bigg). 
\end{align*}
Therefore, it suffices to prove
\begin{align}
    \label{eq_prove_maj_2j+1}
    \MoveEqLeft s \bc{ \Prob{ \Bin \bc{ 2j, \frac{s}{n} } \leq j-1 } + \frac{1}{2} \Prob{ \Bin \bc{ 2j, \frac{s}{n} } } } \\
    \notag & \leq n \bc{ \Prob{ \Bin \bc{ 2j, \frac{s-1}{n} } \leq j-1 } + \frac{1}{2} \Prob{ \Bin \bc{ 2j, \frac{s-1}{n} } = j } } \\
    \notag & \phantom{{}\leq{}} + (s-1) \Big( 1 - \Prob{ \Bin \bc{ 2j, \frac{s-1}{n} } \leq j-1 } - \frac{1}{2} \Prob{ \Bin \bc{ 2j, \frac{s-1}{n} } = j } \Big).
\end{align}
Clearly,
\begin{align*}
    1 - \Prob{ \Bin \bc{ 2j, \frac{s-1}{n} } \leq j-1 } - \frac{1}{2} \Prob{ \Bin \bc{ 2j, \frac{s-1}{n} } = j }  > 0
\end{align*}
by the definition of the Binomial distribution, and as $s \leq n$ we find by the monotonicity of the Binomial distribution
\begin{align*}
    \MoveEqLeft s \bc{ \Prob{ \Bin \bc{ 2j, \frac{s}{n} } \leq j-1 } + \frac{1}{2} \Prob{ \Bin \bc{ 2j, \frac{s}{n} } } } \\
    &\leq n \bc{ \Prob{ \Bin \bc{ 2j, \frac{s-1}{n} } \leq j-1 } + \frac{1}{2} \Prob{ \Bin \bc{ 2j, \frac{s-1}{n} } = j } }.
\end{align*}
Therefore, \eqref{eq_prove_maj_2j+1} follows.
\end{proof}

\subsection{Proof of \texorpdfstring{\cref{lem:diffPoss2jvs2jplus1}}{Lemma 7}}\label{proof_lem_diffPoss2jvs2jplus1}
\diffPossTwojvsTwojplusOne*

\begin{proof}
To compare the processes $\Pn{2j+1}$ and $\Pn{2j}$ in the other cases, we couple the processes such that the first $2j-1$ draws are the same in both processes and for the process $\Pn{2j}$ we draw one more agent in a second step. 
Depending on how many agents with majority opinion have been drawn in the first step, the results of the two processes might differ. 
We summarized the cases that can occur in a decision tree, see \cref{fig:2jvs2jplus1}.
The cases where the two process have a different outcome are highlighted in red and green.

If we draw $j$ agents with majority opinion in the first step, decide to keep the non-majority opinion in the first process, and draw one more majority agent in the last step, 
$x_t^{(2j+1)}$ changes to $a$ in process $\Pl{2j+1}$ but not in process $\Pl{2j}$.
The probability that this case occur is given by 
\[\binom{2j}{j} \alpha^j (1-\alpha)^{j} \cdot \frac{1}{2} \cdot \alpha = \binom{2j}{j} \frac{\alpha^{j+1} (1-\alpha)^{j}}{2}.\]
On the other hand, if we draw $j$ agents with majority opinion, decide to change to the majority opinion in $\Pl{2j}$, and draw one more agent with non-majority opinion in the second step, $x_t^{(2j+1)}$ changes to $a$ for process $\Pl{2j}$ but not for process $\Pl{2j+1}$.
The probability that this case occur is given by 
\[\binom{2j}{j} \alpha^{j} (1-\alpha)^{j}\cdot \frac{1}{2} \cdot (1-\alpha) = \binom{2j}{j} \frac{\alpha^{j} (1-\alpha)^{j+1}}{2}.\]
As a consequence, it holds that
\begin{align*}
\MoveEqLeft \Prob{x_{t+1}^{(2j+1)} = a | x_t^{(2j+1)} =b,   X^{(2j+1)}_{t}= s}\\  
& = \Prob{x_{t+1}^{(2j)} = a | x_t^{(2j)} =b,  X^{(2j)}_{t}= s} +\binom{2j}{j} \frac{\alpha^{j+1} (1-\alpha)^{j}}{2} 
        -\binom{2j}{j} \frac{\alpha^{j} (1-\alpha)^{j+1}}{2}\\
        & = \Prob{x_{t+1}^{(2j)} = a | x_t^{(2j)} =b,  X^{(2j)}_{t}= s} + \frac{ (2\alpha -1)}{2}\binom{2j}{j} \alpha^{j} (1-\alpha)^{j}.
\end{align*}
Note that since $\alpha \geq \frac{1}{2}$ because $s\geq n/2$ it holds that $(2\alpha -1) \geq 0$.

Similarly, when comparing the probabilities for an agent with majority opinion $a$ to change to the non-majority opinion  $b$, we deduce from the tree the following:
The probability that $x_t^{(2j+1)}$ changes to $b$ in $\Pl{2j+1}$ but not for process $\Pl{2j}$ is given by 
\[\binom{2j}{j} \alpha^{j} (1-\alpha)^{j}\cdot \frac{1}{2} \cdot (1-\alpha) = \binom{2j}{j} \frac{\alpha^{j} (1-\alpha)^{j+1}}{2} .\]
On the other hand, the probability that $x_t^{(2j+1)}$ stays  $a$ in $\Pl{2j+1}$ but not in process $\Pl{2j}$ is given by 
\[\binom{2j}{j} \alpha^j (1-\alpha)^{j} \cdot \frac{1}{2} \cdot \alpha = \binom{2j}{j} \frac{\alpha^{j+1} (1-\alpha)^{j}}{2}.\]
As a consequence, it holds that 
\begin{align*}
\MoveEqLeft \Prob{x_{t+1}^{(2j+1)} \rightarrow b | x_t^{(2j+1)} =a,   X^{(2j+1)}_{t}= s} \\ 
= &  \Prob{x_{t+1}^{(2j)} \rightarrow b | x_t^{(2j)} =a,  X^{(2j)}_{t}= s} - \binom{2j}{j} \frac{\alpha^{j} (1-\alpha)^{j+1}}{2} + \binom{2j}{j} \frac{\alpha^{j+1} (1-\alpha)^{j}}{2}\\
        = & \Prob{x_{t+1}^{(2j)} \rightarrow b | x_t^{(2j)} =a,  X^{(2j)}_{t}= s} - \frac{(2\alpha -1)}{2}\binom{2j}{j} \alpha^{j} (1-\alpha)^{j}. \tag*{\qedhere}
\end{align*}

\end{proof}

\subsection{Proof of \texorpdfstring{\cref{lem:differentProbabilities}}{Lemma 8}}\label{proof_lem_differentProbabilities}
\differentProbabilities*
\begin{proof}
Note that for $d = s-1$ it holds that
\begin{align*}
\Prob{X^{(2j+1)}_{t+1} \geq d | X^{(2j+1)}_{t} = s} = 1 =  \Prob{X^{(2j)}_{t+1}\geq d | X^{(2j)}_{t}= s}
\end{align*}
as well as for $d=s+2$ it holds that
\begin{align*}
\Prob{X^{(2j+1)}_{t+1} \geq d | X^{(2j+1)}_{t} = s} = 0 =  \Prob{X^{(2j)}_{t+1}\geq d | X^{(2j)}_{t}= s}.
\end{align*}

Let us first compare the probabilities that the processes increase the number of agents with majority opinion in the next step. 
Let $\alpha = \frac{s}{n}$.
For the majority to increase the agent which is updated has to be an agent with non-majority opinion. 
The probability to draw an agent with non-majority opinion is given by $(1-\alpha)$.
As a consequence, by \cref{lem:diffPoss2jvs2jplus1} it holds that
\begin{align*}
    \MoveEqLeft \Prob{X^{(2j+1)}_{t+1} = s+1 | X^{(2j+1)}_{t} = s} \\
    &= (1-\alpha)\Prob{x_{t+1}^{(2j+1)} = a | x_t^{(2j+1)} =b,   X^{(2j+1)}_{t}= s}\\
    &= \Prob{X^{(2j)}_{t+1} = s+1 | X^{(2j)}_{t} = s} + \frac{(2\alpha -1)}{2}\binom{2j}{j} \alpha^{j} (1-\alpha)^{j+1}.
\end{align*}
As a consequence, for $d = s+1$ it holds that
\begin{align*}
\Prob{X^{(2j+1)}_{t+1} \geq d | X^{(2j+1)}_{t} = s} \geq \Prob{X^{(2j)}_{t+1} \geq d | X^{(2j)}_{t} = s} .
\end{align*}
Similarly, the number of majority agents only can decrease if the updated agent has the majority opinion. 
This happens with probability $\alpha$ and hence
\begin{align*}
    \MoveEqLeft \Prob{X^{(2j+1)}_{t+1} = s-1 | X^{(2j+1)}_{t} = s} \\
    &= \alpha\Prob{x_{t+1}^{(2j+1)} = b | x_t^{(2j+1)} =a,   X^{(2j+1)}_{t}= s}\\
    &= \Prob{X^{(2j)}_{t+1} = s-1 | X^{(2j)}_{t} = s} - \frac{(2\alpha -1)}{2}\binom{2j}{j} \alpha^{j+1} (1-\alpha)^{j}.
\end{align*}

Since for any process it holds that
\begin{align*}
\Prob{X^{(k)}_{t+1} \geq s | X^{(k)}_{t} = s} = 1- \Prob{X^{(k)}_{t+1} = s-1 | X^{(k)}_{t} = s}
\end{align*}
and it holds that
\begin{align*}
\Prob{X^{(2j+1)}_{t+1} = s-1 | X^{(2j+1)}_{t} = s}
\leq \Prob{X^{(2j)}_{t+1} = s-1 | X^{(2j)}_{t} = s}
\end{align*}
we get for $d = s$ that
\begin{align*}
{\Prob{X^{(2j+1)}_{t+1} \geq d | X^{(2j+1)}_{t} = s} \geq \Prob{X^{(2j)}_{t+1} \geq d | X^{(2j)}_{t} = s}.} \tag*{\qedhere}
\end{align*}
\end{proof}

\subsection{Proof of \texorpdfstring{\cref{lem:diffPoss2jminus1vs2j}}{Lemma 9}}
\label{proof_lem_diffPoss2jminus1vs2j}
\diffPossTwojminusOnevsTwoj*
\begin{proof}
To compare the processes $\Pn{2j-1}$ and $\Pn{2j}$ in the other cases, we couple the processes such that the first $2j-1$ draws are the same in both processes and for the process $\Pn{2j}$ we draw one more agent in a second step. 
Depending on how many agents with majority opinion have been drawn in the first step, the results of the two processes might differ. 
We summarized the cases that can occur in a decision tree, see \cref{fig:2jvs2jplus1}.
The cases where the two process have a different outcome are highlighted in red and green.

If we draw $j-1$ agents with majority opinion in the first step, draw one more majority agent in the second step, and decide with probability $1/2$ to change the opinion of the agent,
$o_t^{(k)}(x)$ changes to $a$ in process $\Pl{2j}$ but not in process $\Pl{2j-1}$.
The probability that this case occur is given by 
\[\binom{2j-1}{j-1} \alpha^{j-1} (1-\alpha)^{j} \cdot \alpha \cdot \frac{1}{2} = \binom{2j-1}{j} \frac{\alpha^{j} (1-\alpha)^{j}}{2} .\]

On the other hand, if we draw $j$ agents with majority opinion, draw one more agent with non-majority opinion in the second step, and decide with  probability $1/2$ to change the opinion of the agent, 
$o_t^{(k)}(x)$ changes to $a$ for $k={2j-1}$ but not for $k={2j}$.
The probability that this case occur is given by 
\[\binom{2j-1}{j} \alpha^{j} (1-\alpha)^{j-1} \cdot (1-\alpha) \cdot \frac{1}{2} = \binom{2j-1}{j} \frac{\alpha^{j} (1-\alpha)^{j}}{2}.\]
As a consequence, it holds that
\begin{align*}
\MoveEqLeft \Prob{x_{t+1}^{(2j)} = a | x_t^{(2j)} =b,   X^{(2j)}_{t}= s} \\
&= \Prob{x_{t+1}^{(2j-1)} = a | x_t^{(2j-1)} =b,  X^{(2j-1)}_{t}= s}\\
        & \phantom{{}={}} +\binom{2j-1}{j} \frac{\alpha^{j} (1-\alpha)^{j}}{2} -\binom{2j-1}{j} \frac{\alpha^{j} (1-\alpha)^{j}}{2}\\
        & =  \Prob{x_{t+1}^{(2j-1)} = a | x_t^{(2j-1)} =b,  X^{(2j-1)}_{t}= s} .
\end{align*}
Note that since $\alpha \geq 0.5$ because $s\geq n/2$ it holds that $(2\alpha -1) \geq 0$.

Similarly, when comparing the probabilities for an agent with majority opinion $a$ to change to the non-majority opinion  $b$, we deduce from the tree the following:
The probability that $o_t^{(k)}(x)$ changes to $b$ for $k=2j$ but not for $k=2j-1$ is given by
\[\binom{2j-1}{j} \alpha^{j} (1-\alpha)^{j-1} \cdot (1-\alpha) \cdot \frac{1}{2}= \binom{2j-1}{j} \frac{\alpha^{j} (1-\alpha)^{j}}{2}.\]
On the other hand, the probability that $o_t^{(k)}(x)$ changes to $b$ for $k=2j-1$ but not for $k=2j$ is given by 
\[\binom{2j-1}{j-1} \alpha^{j-1} (1-\alpha)^{j} \cdot \alpha \cdot \frac{1}{2}= \binom{2j-1}{j} \frac{\alpha^{j} (1-\alpha)^{j}}{2}.\]
As a consequence, it holds that 
\begin{align*}
\MoveEqLeft \Prob{x_{t+1}^{(2j)} \rightarrow b | x_t^{(2j)} =a,   X^{(2j)}_{t}= s}  \\
&=  \Prob{x_{t+1}^{(2j-1)} \rightarrow b | x_t^{(2j-1)} =a,  X^{(2j-1)}_{t}= s}\\
& \phantom{{}={}} + \binom{2j-1}{j} \frac{\alpha^{j} (1-\alpha)^{j}}{2} - \binom{2j-1}{j} \frac{\alpha^{j} (1-\alpha)^{j}}{2}\\
        &= \Prob{x_{t+1}^{(2j-1)} \rightarrow b | x_t^{(2j-1)} =a,  X^{(2j-1)}_{t}= s}.\tag*{\qedhere}
\end{align*}
\end{proof}

\subsection{Proof of \texorpdfstring{\cref{lem:differentProbabilities2}}{Lemma 10}}
\differentProbabilitiesTwo*
\begin{proof}
To compare the processes $\Pn{2j-1}$ and $\Pn{2j}$ in the other cases, we couple the processes such that the first $2j-1$ draws are the same in both processes and for the process $\Pn{2j}$ we draw one more agent in a second step. 
Depending on how many agents with majority opinion have been drawn in the first step, the results of the two processes might differ. 
We summarized the cases that can occur in a decision tree, see \cref{fig:2jplus1vs2jplus2}.
The cases where the two process have a different outcome are highlighted in red and green.

Let us first compare the probabilities that the processes increase the number of agents with majority opinion in the next step. 
Let $\alpha = \frac{s}{n}$.
For the majority to increase the agent which is updated has to be an agent with non-majority opinion. 
The probability to draw an agent with non-majority opinion is given by $(1-\alpha)$.
As a consequence, by \cref{lem:diffPoss2jminus1vs2j} it holds that
\begin{align*}
    \MoveEqLeft \Prob{X^{(2j)}_{t+1} = s+1 | X^{(2j)}_{t} = s} \\
    &= (1-\alpha)\Prob{x_{t+1}^{(2j)} = a | x_t^{(2j)} =b,   X^{(2j+1)}_{t}= s}\\
    &= (1-\alpha)\left(\Prob{x_{t+1}^{(2j-1)} = a | x_t^{(2j-1)} =b,  X^{(2j-1)}_{t}= s}\right)\\
    &= \Prob{X^{(2j-1)}_{t+1} = s+1 | X^{(2j-1)}_{t} = s} + \frac{(2\alpha -1)}{2}\binom{2j}{j} \alpha^{j} (1-\alpha)^{j+1}.
\end{align*}
Similarly, the number of majority agents only can decrease if the updated agent has the majority opinion. 
This happens with probability $\alpha$ and hence, by
\cref{lem:diffPoss2jminus1vs2j}, we get
\begin{align*}
    \MoveEqLeft \Prob{X^{(2j)}_{t+1} = s-1 | X^{(2j)}_{t} = s} \\
    &= \alpha\Prob{x_{t+1}^{(2j)} = b | x_t^{(2j)} =a,   X^{(2j)}_{t}= s}\\
    &= \alpha\Prob{x_{t+1}^{(2j-1)} = b | x_t^{(2j-1)} =a,  X^{(2j-1)}_{t}= s}\\
    &= \Prob{X^{(2j-1)}_{t+1} = s-1 | X^{(2j-1)}_{t} = s}
\end{align*}     
Therefore we have for any $d \in [n]$ that
\begin{align*}
\Prob{X^{(2j)}_{t+1} \geq d | X^{(2j)}_{t} = s}    & = \Prob{X^{(2j-1)}_{t+1} \geq d | X^{(2j-1)}_{t} = s}.\tag*{\qedhere}
\end{align*}
\end{proof}

\end{document}